\documentclass[11pt]{article}

\usepackage{lineno,hyperref}
\modulolinenumbers[5]

\usepackage{geometry}
\usepackage[dvipsnames]{xcolor}
\usepackage{color}

\newcommand{\update}[1]{\textcolor{Black}{#1}}

\usepackage{amssymb}
\usepackage{amsmath}
\usepackage{amsthm} 
\usepackage{amsfonts}

\usepackage{graphicx}
\usepackage{amsmath}
\usepackage{amsthm}
\usepackage{amssymb}
\usepackage{tikz-cd}

\usepackage{marginnote}
\usepackage{todonotes} 

\usepackage{hyperref}
\usepackage{color}

\usepackage[capitalize]{cleveref}

\newtheorem{thm}{Theorem}[section]

\newtheorem{lem}[thm]{Lemma}
\newtheorem{prop}[thm]{Proposition}
\newtheorem{cor}[thm]{Corollary}
\newtheorem{claim}[thm]{Claim}
\newtheorem{definition}[thm]{Definition}

\newcommand{\ds}{\displaystyle}

\definecolor{darkred}{rgb}{1, 0.1, 0.3}
\definecolor{darkblue}{rgb}{0.1, 0.1, 1}
\definecolor{darkgreen}{rgb}{0,0.6,0.5}

\newcommand{\denselist}{\itemsep 0pt\parsep=1pt\partopsep 0pt}
\newcommand{\Image}		{\mathrm{Im}}

\newcommand{\Mmid}  {{M_{mid}}}

\DeclareMathOperator*{\LCA}{LCA}

\newcommand{\R}{\mathbb{R}}
\newcommand{\X}{\mathbb{X}}
\newcommand{\Rspace}{\mathbb{R}}
\newcommand{\Xspace}{\mathbb{X}}

\newcommand{\inv}{^{-1}}
\newcommand{\LMT}{\mathsf{LMT}}

\newcommand{\MT}{\mathsf{MT}}
\newcommand{\UM}{\mathsf{UM}}
\newcommand{\VM}{\mathsf{VM}}

\newcommand{\M}{\mathcal{M}}
\newcommand{\MM}{\mathcal{M}}

\newcommand{\TT}{\mathcal{T}}
\newcommand{\UU}{\mathcal{U}}

\newcommand{\depth}{\mathrm{depth}}

\DeclareMathOperator*{\argmin}{arg\,min}
\title{Intrinsic Interleaving Distance for Merge Trees}
\author{Ellen Gasparovic\thanks{Union College, gasparoe@union.edu}, Elizabeth Munch\thanks{Michigan State University, muncheli@msu.edu}, Steve Oudot\thanks{Inria Saclay, steve.oudot@inria.fr}, \\Katharine Turner\thanks{Australian National University, katharine.turner@anu.edu.au}, Bei Wang\thanks{University of Utah, beiwang@sci.utah.edu}, Yusu Wang\thanks{Ohio State University, yusu@cse.ohio-state.edu}}
\date{}

\begin{document}

%\begin{frontmatter}

\title{Intrinsic Interleaving Distance for Merge Trees}

%\author[1]{Ellen Gasparovic}\corref{cor1}
%\author[2]{Elizabeth Munch}
%\author[3]{Steve Oudot}
%\author[4]{Katharine Turner}
%\author[5]{Bei Wang}
%\author[6]{Yusu Wang}
%\address[1]{Union College, gasparoe@union.edu}
%\address[2]{Michigan State University, muncheli@msu.edu} 
%\address[3]{Inria Saclay, steve.oudot@inria.fr}
%\address[4]{Australian National University, katharine.turner@anu.edu.au}
%\address[5]{University of Utah, Salt Lake City, UT, USA, beiwang@sci.utah.edu}
%\address[6]{University of California, San Diego, CA, USA, yusuwang@ucsd.edu}

%\begin{keyword}
% merge tree \sep interleaving distance \sep intrinsic distance 
%\end{keyword}

%\end{frontmatter}

\pagestyle{empty}

\maketitle 

\begin{abstract} 
Merge trees are a type of graph-based topological summary that track the evolution of connected components in the sublevel sets of scalar functions. They enjoy widespread applications in data analysis and scientific visualization. 
In this paper, we consider the problem of comparing two merge trees via the notion of interleaving distance in the metric space setting. 
We investigate various theoretical properties of such a metric. 
In particular, we show that the interleaving distance is intrinsic on the space of labeled merge trees and provide an algorithm to construct metric 1-centers for collections of labeled merge trees. 
We further prove that the intrinsic property of the interleaving distance also holds for the space of unlabeled merge trees. 
Our results are a first step toward performing statistics on graph-based topological summaries. 
\end{abstract}

\pagestyle{plain}

%-------------------------
% Introduction
\section{Introduction}
\label{sec:introduction}

Many applications in science and engineering use scalar functions to describe and model their data. 
For example, atmospheric scientists compare simulated data from the Weather Research and Forecasting (WRF) Model with daily surface observations in weather forecasts, where both simulated and observed parameters (such as surface temperature, pressure, precipitation, and wind speed) can be modeled as scalar functions. 
We are interested in comparing scalar functions by comparing their topological summaries. 
There are several types of summaries constructed from topological methods, including vector-based summaries such as persistence diagrams~\cite{EdelsbrunnerLetscherZomorodian2002} and  barcodes~\cite{Ghrist2008}, as well as graph-based summaries such as merge trees, contour trees~\cite{CarrSnoeyinkAxen2003}, and Reeb graphs~\cite{Reeb1946}. 

The merge tree (sometimes referred to as a barrier tree~\cite{FlammHofackerStadler2002} or a join  tree~\cite{CarrSnoeyinkAxen2003}) for a given topological space $\Xspace$ \update{equipped with a continuous scalar function is a combinatorial construction that tracks the evolution of sublevel sets}.  
For a given function $f: \Xspace \to \Rspace$, the merge tree encodes the connected components of the sublevel sets $f\inv(-\infty,a]$ for $a \in \Rspace$. 
This construction is closely related to that of the Reeb graph~\cite{Reeb1946}, which analogously encodes connected components of the level sets $f\inv(a)$. The contour tree~\cite{CarrSnoeyinkAxen2003} is a special type of Reeb graph for a simply connected domain. 
Both merge trees and Reeb graphs are related to the level set topology through critical points of the scalar functions~\cite{Milnor1963}. 
 Furthermore, the mapper graph~\cite{SinghMemoliCarlsson2007}, which has found considerable success in applications, can be viewed as an approximation of a Reeb graph~\cite{CarriereOudot2018,MunchWang2016,CarriereMichelOudot2018}. 
 These constructions are referred to as graph-based summaries as the output object of study is always a graph $G$ equipped with an induced real-valued function $f:G \to \Rspace$.  
 They have appeared in many contexts and applications over the last few decades~\cite{WeberBremerPascucci2007,OesterlingHeineJaenicke2011,WidanagamaachchiJacquesWang2017}. 
 \update{Similar concepts also appeared within probability theory as trees created through excursion sets of random functions, and these trees are shown to be related to random branching processes (e.g.~\cite{LeGall1991,Evans2006})}. 

\paragraph{\update{Related work}}
Considerable recent effort has gone into understanding how to perform proper statistics on graph-based summaries.  
For instance, how does one define the mean of a collection of these objects? 
The first step toward answering this question is to determine a metric for the comparison of two summaries.  
This has been extensively studied recently with the creation of a veritable zoo of metric options for Reeb graphs and merge trees~\cite{Morozov2013,deSilva2016,Bauer2014, Bauer2015b,BauerDiFabioLandi2016,DiFabio2016,BauerLandiMemoli2017,Sridharamurthy2018, CarriereOudot2017, Beketayev2014}; see two recent surveys~\cite{YanMasoodSridharamurthy2021,BollenChambersLevine2021} and \cref{ssec:AvailableMetrics} for a discussion of some of these metrics.  
In particular, Carri\'ere and Oudot~\cite{CarriereOudot2017} have investigated whether some of these metrics are \emph{intrinsic} in the more general case of Reeb graphs; i.e.,~that the distance between two (close enough) graphs can be realized by a geodesic. 

In this paper, we continue the investigation into the intrinsic-ness of these metrics with the more narrow view of merge trees.  
The main distance we study is the interleaving distance.  
This metric was originally given in the context of persistence modules~\cite{Chazal2009b,Chazal2016} as a generalization of the bottleneck distance, and has been ported to merge trees~\cite{Morozov2013,Touli2018} and Reeb graphs~\cite{deSilva2016,Curry2014} via a category-theoretic viewpoint~\cite{Bubenik2014a, deSilva2018}.
When restricting ourselves from Reeb graphs to merge trees, we can actually work in an even more restrictive setting that has desirable theoretical properties, namely, labeled merge trees. 
In this case, we study a data triple: a merge tree $T$ with its function $f:T \to \R$, and a labeling $\pi:\{1,\cdots, n\} \to V(T)$ of its vertices, which at a minimum encompasses the leaves of $T$.  
The interleaving distance for labeled merge trees has been investigated in~\cite{Munch2018}, where it is shown that the metric can be naturally realized as the $L_\infty$-distance for a particular matrix construction.  
This \update{matrix construction} has already been discovered in the context of dendrograms \cite{Sokal1962} and phylogenetic trees \cite{Cardona2013}, where the objects of interest are closely related to merge trees.  
The phylogenetic tree literature, in particular, provides a wealth of other options for metrics~\cite{Robinson1979,Robinson1981, DasGupta1997,DasGupta1999,Diaconis1998,Billera2001,Bogdanowicz2013, Bogdanowicz2012e,Estabrook1985, Cardona2009,Choi2009,Lafond2019}. 
There has also been interest in that community for creating summaries of collections of phylogenetic trees~\cite{Miller2015, Markin2017, Gavryushkin2016}. 

These ideas are also closely related to those of ultrametrics, a strengthening of the triangle inequality for a metric into a requirement that $d(x,y) \leq \max\{d(x,z), d(z,y) \}$ (\update{for all $z$}). 
Independent of the phylogenetic tree work, there has been extensive interest in what is known as Gelfand's Problem from the ultrametric literature, that is, to describe all finite ultrametric spaces up to isometry using graph theory.  
The answer to this question is exactly a restriction of the labeled merge tree, although their literature never calls it such~\cite{Gurvich2012, Hughes2004, Lemin2003, Dovgoshey2016, Dovgoshey2018}.

\update{Furthermore, our work has close ties with the literature on consensus of classification~\cite{Leclerc1998}. In the language used therein, labeled merge trees belong to the class of valued classification trees, with the path-length distance being actually induced by the function values. Geodesic midpoints between two merge trees are called {\em medians} and defined as Fr\'echet means in the considered metric between trees. Finding medians is a special instance of the so-called consensus problem, and it is known to have an easily computable solution 
when the metric between merge trees is chosen to be the $\ell^\infty$-distance between their corresponding ultra matrices~\cite{ChepoiFichet}, as is the case in our work. By contrast, the problem is known to be NP-hard when other $\ell^p$-distances are put on the ultra matrices, typically when $p=1$ or~$2$~\cite{barthelemy1993median}, and in such situations one must resort to approximate solutions---e.g., mean-squares approximations in the case $p=2$~\cite{deSoete1983least,lapointe1997average}.}

\paragraph{Our contributions} In this paper, we prove that the interleaving distance is intrinsic on both the space of labeled merge trees as well as on the space of unlabeled merge trees.
Furthermore, in the labeled merge tree setting, we provide explicit procedures for constructing geodesics and metric 1-centers. \update{Our results mark important progress toward the goal of performing statistics on graph-based topological summaries to be used for topological data analysis and visualization. For instance, using results in this paper, Yan et al.~\cite{YanWangMunch2020}  computed geodesics of merge trees and their structural averages for ensemble analysis and  uncertainty visualization, and Curry et al.~\cite{CurryHangMio2021} utilized an estimation of the interleaving distance between unlabeled merge trees in order to classify and compare point cloud data.}

In Section~\ref{sec:Background}, we provide the necessary background on labeled merge trees and establish a correspondence between labeled merge trees and a particular class of matrices known as ultra matrices (Lemma~\ref{lem:UM_iso_LMT}). 
Then, in Section~\ref{sec:labeled}, we prove a stability result for the labeled interleaving distance $d_I^L$ (Lemma~\ref{lem:ultrammt-Lip}), which we use both to show that $d_I^L$ is \emph{strictly intrinsic} on the space of labeled merge trees (Corollary~\ref{cor:labelintrinsic}) as well as to construct 1-centers for collections of labeled merge trees in Section~\ref{subsec:1centers}.

Section~\ref{sec:unlabeled} focuses on unlabeled merge trees and the interleaving distance.
In particular, given two unlabeled merge trees, we show that the unlabeled interleaving distance between them is equal to the infimum over all finite labelings for the two trees of the labeled interleaving distance between them (Theorem~\ref{thm:LabeledVsUnlabeledInterleaving}). In fact, we show in Corollary~\ref{cor:OptLabel} that the infimum is always achieved. Section~\ref{sec:unlabeled} concludes with the result that the interleaving distance is intrinsic on the space of unlabeled merge trees (Corollary~\ref{cor:intrinsicinterleaving}). 
We end with a discussion of open problems and future work in Section~\ref{sec:discussion}.

%-------------------------
% Background

\section{Background}
\label{sec:Background}

In this section, we give the basic definitions for our constructions of interest.  
We refer to \cref{fig:Roadmap} for an overview of notations.
For the entirety of the section, we fix $n$, and denote $\{1,\cdots,n\}$ by $[n]$ and isomorphism by $\cong$. 
\update{We note that some of the notions here appear in the literature under different names. For example, the concept of ultra matrix is the same as the one induced from the \emph{ultra network} proposed in \cite{SCM16}, both of which correspond to the distance matrix for a relaxed version of the ultrametric (see~\cref{def:relaxedUltrametric}). 
The concept of a merge tree is also the same as the \emph{tree gram} in \cite{SCM16}, both of which generalize the standard dendrogram. 
A dendrogram could be represented as an ultrametric, as shown by Jardine and Sibson~\cite{JardineSibson1971}, Hartigan~\cite{Hartigan1985}, Carlsson and M\'{e}moli~\cite{CM10}.}

\begin{figure}[!h]
    \centering
    $
    \begin{tikzcd}
    \text{Valid Matrices} = \VM \ar[dr, "\TT", bend left] \ar[dd, "\UU = \MM \TT"', bend left = 80]  \\
     & \LMT = \text{Labeled Merge Trees}  \ar[dl, "\MM", "\cong"', bend left  ] \ar[r] & \MT = \text{Merge trees} \\
    \text{Ultra Matrices} = \UM   \ar[uu, hook] 
    \end{tikzcd}
    $
    \caption{A roadmap of key notations.}
    \label{fig:Roadmap}
\end{figure}
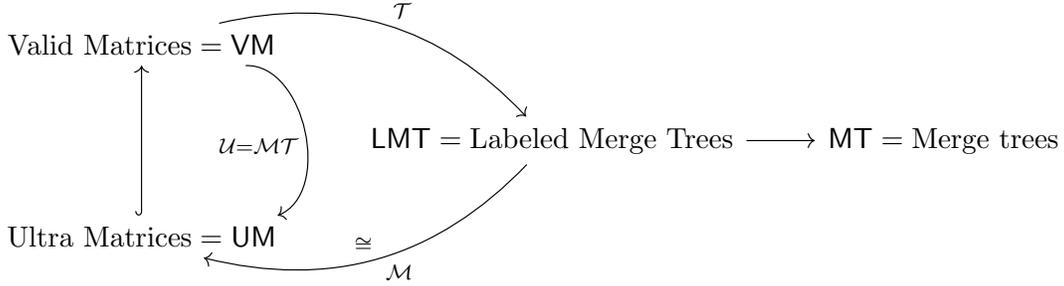

\subsection{Labeled Merge Trees} 

\update{First, we give the definition of a \emph{merge tree} (which we shall also refer to as an \emph{unlabeled merge tree} to contrast it with its labeled counterpart, defined subsequently) and related notions arising from the phylogenetic tree literature that we will make use of shortly. }

\begin{definition}
\label{def:mergeTree}
A \textbf{merge tree} is a  pair $(T,f)$ of a finite rooted tree $T$ with vertex set $V(T)$ and a function $f:V(T) \to \R\cup \infty$ such that adjacent vertices do not have equal function value, every non-root vertex has exactly one neighbor with higher function value, and the root (a degree one node) is the only point with the value $\infty$.
The space of merge trees is denoted $\MT$. 
\end{definition}
\noindent \update{Note that the topology on the space of merge trees is the one induced by viewing $\MT$ as a metric space with a given choice of metric (in this paper, we shall assume it is the interleaving distance, as defined later in this section).} We commonly call the function $f$ a \emph{height function}, the non-root vertices with degree 1 are called \emph{leaves}, and we let $depth(u)$ denote the largest height difference between the vertex $u$ in $T$ and any node in the subtree rooted at $u$. 
All merge trees under consideration in this paper are assumed to be   \emph{finite}.   

\update{
Note that in some cases --- for instance, when a merge tree is constructed from sublevel sets of input data given by a topological space with a function $g:\X \to \R$ --- we prefer to think of a merge tree as a continuous object. 
Specifically, a merge tree is constructed as a quotient space $\X/\sim$, with the equivalence relation $x \sim x'$ iff $g(x) = g(x') = a$ and $x$ and $x'$ are in the same connected component of $g\inv(\infty,a]$. 
In this case, with nice enough assumptions on the input data, a merge tree is a 1-dimensional stratified space equipped with a function, that is, its edges are homeomorphic to intervals rather than being combinatorial elements.} 
In this setting, replacing a merge tree edge $e = (u,v)$ with $f(u) < f(v)$ by a subdivision of that edge where the interior vertex $w$ satisfies $ f(u) < f(w) <f(v) $ does not change the inherent structure of the tree (sometimes such a tree is referred to as an \emph{augmented merge tree}). 
We consider two merge trees to be the same if one can be obtained from the other by a sequence of such subdivisions or the inverse operation. 

Furthermore, the merge tree structure induces a poset relation on the vertices of $T$. 
We say $v$ is an \emph{ancestor} of $w$ and write $v \succ w$ if the unique path from $v$ to $w$ strictly decreases in $f$. 
This occurs if and only if $w$ is in the subtree of $v$. 
We use $\LCA(v,w) \in T $ to mean the lowest common ancestor of $v$ and $w$ (or $\LCA_f(v,w)$ if the function needs to be emphasized), and $f(\LCA(v, w) )$ for its function value. 
We have $\LCA(v,v) = v.$
We abuse notation and write $\LCA(S)$ for the lowest common ancestor of any finite set $S \subset V(T)$.

\update{
Note that the merge tree as defined is closely related to the construction of a rooted, weighted tree. 
In fact, a merge tree induces a rooted weighted tree by putting the weight $f(u) - f(v)$ on each directed edge of the tree $T$. 
However, because of the function setting, the merge tree requirements are stricter since, for instance, a merge tree $(T,f)$ and its shift (i.e.,~translation) $(T,f+100)$ are considered different as merge trees but induce the same weighting. 
}

The merge tree structure provides a method for inducing a metric on the underlying tree vertices via the metric given by the length of the unique path between two points.
Note that there is a canonical weighting associated to any merge tree $(T,f)$, namely, $\omega(u,v) = |f(u) - f(v)|$ for any two adjacent vertices $u$ and $v$ in the tree. 
Furthermore, as paths are unique in a tree, we can define a metric for any pair of vertices by $\delta_T(u,v) = \sum \omega(e)$ for the edges in the path from $u$ to $v$. 

\update{We remind the reader that we use the terms \emph{merge tree} and \emph{unlabeled merge tree} interchangeably.} In~\cref{sec:labeled}, we will be focusing on \emph{labeled merge trees}, defined as follows.

\begin{definition}
\label{def:labeledMergeTree}
A \textbf{labeled merge tree} is a triple $(T,f,\pi)$ consisting of a merge tree $(T,f)$ along with a map $\pi:[n] \to V(T)$ that is surjective on the set of leaves.
When additional data are unnecessary or clear from context, we sometimes write $T$ for $(T,f,\pi)$.
The space of labeled merge trees is denoted $\LMT$.
\end{definition}

\update{Note that the topology on $\LMT$ comes from viewing it as a metric space with the labeled interleaving distance, as defined later in this section.} Analogous to the unlabeled case, we consider two labeled merge trees to be the same if one can be obtained via edge contractions or insertions that respect the function values and existing labels. 

\cref{def:labeledMergeTree} is closely related to that of a \textit{weighted, rooted $X$-tree} from the phylogenetic literature \cite{Semple2003}. 
Specifically, given a set $X$, an $X$-tree is a pair $(T,\phi)$ where $T$ is a tree and $\phi:X \to V(T)$ is a map so that every vertex of degree at most 2 is in the image. 
The difference is that such weighted graphs do not keep track of function values, so that two different labeled merge trees that induce the same weighting might be considered to be the same $X$-tree. 
Thus, a labeled merge tree can be thought of as a weighted, labeled $X$-tree (where $X = [n]$) with $f(u)$ specified for a subset of vertices $u$ that includes all leaves, and function values for the remaining vertices can be deduced from the weights on leaves. 

As with $X$-trees, labels for our merge tree are allowed to go to vertices that are not leaves; we essentially think of these as degenerate labeled leaves.
Furthermore, we do allow $\pi$ to be non-injectivite, so a vertex can have multiple labels. 
See \cref{fig:DegenerateLeaves} for an example with labels on degenerate leaves and vertices with more than one label. 

\begin{figure*}
    \centering
        \begin{minipage}{.18\textwidth}
    	$
    	\begin{pmatrix}
	a_1 & a_1 & a_4 & a_4 \\
	\cdot & a_1 & a_4 & a_4\\
	\cdot & \cdot & a_3 & a_3\\
	\cdot & \cdot & \cdot & a_2
	\end{pmatrix}$
    \end{minipage}
        \begin{minipage}{.2\textwidth}
        $$
        \begin{tikzcd}
         \phantom{x} \ar[r, bend left, "\TT"]  &  \ar[l, bend left, "\MM"] \phantom{x}
        \end{tikzcd}
        $$
    \end{minipage}
    \begin{minipage}{.2\textwidth}
    \includegraphics[width = \textwidth ] {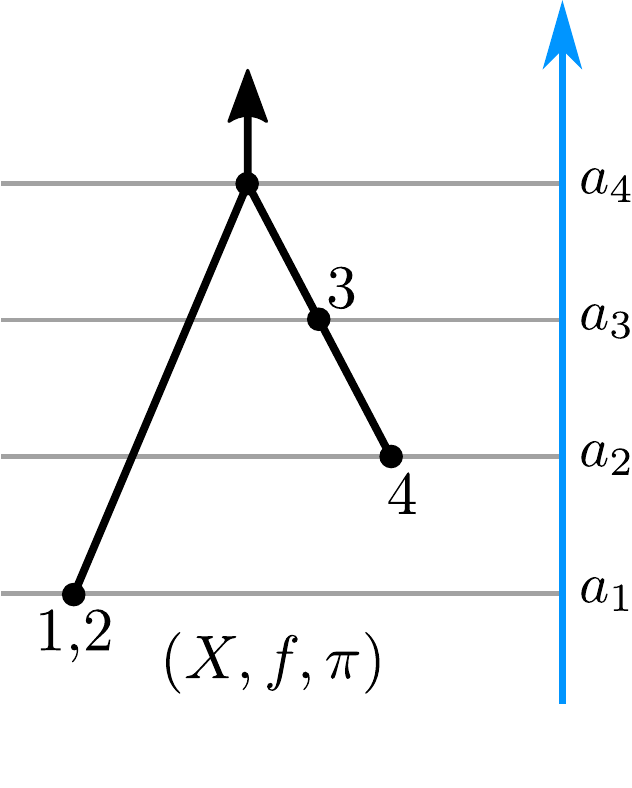}
    \end{minipage}
    \caption{An example of a labeled merge tree with two types of degenerate labels.  As all matrices used are symmetric, we only show the upper triangular portion.}
    \label{fig:DegenerateLeaves}
\end{figure*}

\subsection{Relating Merge Trees and Matrices}

In this section, we give the relationship between labeled merge trees and a particular class of matrices.  
Again, see \cref{fig:Roadmap} for an overview of notation. 

We begin with the traditional notion of an \emph{ultrametric} and our variant of it that relaxes one of the conditions, which will be closely related to our labeled merge trees.

\begin{definition}
\label{def:Ultrametric}
An \textbf{ultrametric} is a function $d:X \times X \to \R$ such that for any $x,y,z \in X$, 
\begin{itemize}\denselist
\item $d(x,y) \geq 0$ and is equal to 0 if and only if $x = y$,
\item $d(x,y) = d(y,x)$, and
\item $d(x,y) \leq \max\{d(x,z), d(z,y)\}$.
\end{itemize}
\end{definition}

\begin{definition}
\label{def:relaxedUltrametric}
A \textbf{relaxed ultrametric} is a function $d:X \times X \to \R$ such that for any $x,y,z\in X$,
\begin{itemize}\denselist
\item $d(x,y) = d(y,x)$, and
\item $d(x,y) \leq \max\{d(x,z), d(z,y)\}$.
\end{itemize}
\end{definition}

It is well known that ultrametrics satisfy the isosceles triangle property.
That is, for any triple $x,y,z$, at least two of $d(x,y)$, $d(y,z)$, and $d(x,z)$ must be equal.  
Otherwise, assume without loss of generality that $d(x,y)< d(y,z)<d(x,z)$, and then $d(x,z) \not \leq \max\{d(x,y),d(y,z) \}$. 
Note that this further implies that the pair that are equal must be at least as big as the third value, since $d(x,y)= d(y,z)<d(x,z)$ still violates the ultrametric property. 
Note that relaxed ultrametrics still satisfy the isosceles triangle property.  

When we have a set $X \cong [n]$, the information in a relaxed ultrametric can be stored as follows. 
\begin{definition}
A symmetric matrix $M \in \R^{n \times n}$ is called \textbf{valid} if $M_{ii} \leq M_{ij} $ for all $1 \leq i,j \leq n$.
A valid matrix $M$ is called \textbf{ultra} if 
$M_{ij} \leq \max \{ M_{ik}, M_{kj}\}$ for every $1\le k\le n$.
The spaces of valid and ultra matrices are denoted $\VM$ and $\UM$, respectively.
\end{definition}

\noindent In particular, a relaxed ultrametric on $[n]$ is represented by an ultra matrix. 
\update{
As with merge trees, we will endow $\VM$ and $\UM$ with the topology induced by the relevant metric, in this case, the $\ell^\infty$ distance between matrices, $\|M-M'\|_\infty = \max_{i,j} |m_{i,j}-m'_{i,j}|$.
}

Inspired by the cophenetic matrix construction of Cardona et al.~\cite{Cardona2013} that is studied in relation to merge trees in \cite{Munch2018}, there is a natural way to associate a matrix to a labeled merge tree as follows.

\begin{definition}
\label{def:cophVector}
The \textbf{induced matrix of a labeled merge tree} $(T,f,\pi)$, denoted $\MM(T,f,\pi) \in \R^{n \times n}$, is the matrix 
\begin{equation*}
\MM(T,f,\pi)_{ij} = f( \LCA( \pi(i),  \pi(j))). 
\end{equation*}
\end{definition}
\noindent See \cref{fig:DegenerateLeaves} for an example.

\begin{lem}
The induced matrix of a labeled merge tree is an ultra matrix.
That is, $\MM(T,f,\pi) \in \UM$ for $(T,f,\pi) \in \LMT$.
\end{lem}

\begin{proof}
Let $M = \MM(T,f,\pi)$ for $(T,f,\pi) \in \LMT$.
First, to check that it is a valid matrix, we see that $M_{ii}$ is simply the function value $f(\pi(i))$.
So, as $f(u) \leq f(LCA(u,v))$ by definition, we have
\begin{equation*}
M_{ii} = f(\pi(i)) \leq f( \LCA(\pi(i),\pi(j))) = M_{ij}.
\end{equation*}
\noindent To check that $M$ is an ultra matrix, let $u = \LCA(\pi(i),\pi(k))$, $v = \LCA(\pi(j),\pi(k))$, $w = \LCA(\pi(i),\pi(j), \pi(k))$.
This means that $u \preceq w$ and $v \preceq w$. 
If $u$ and $v$ are not comparable, then there are two distinct paths from $\pi(j)$ to each of them, and thus we have a loop $\pi(j)\to u\to w\to v\to \pi(j)$, contradicting the tree property of $T$. 
If $u$ and $v$ are comparable, assume without loss of generality that $u \preceq v$; then $v$ is a common ancestor for $\pi(i)$, $\pi(j)$, and $\pi(k)$, and thus $w \preceq v$. 
This implies $f(w) \leq f(v)$, and so for all $k$, 
\[
    M_{ij} \leq f(w) \leq f(v) = \max\{f(u),f(v)\} = \max\{M_{ik}, M_{jk} \}. \qedhere
    \]
\end{proof}

A valid matrix may be viewed as representing a function $f_M$ on a complete graph $K$ of $n$ vertices, with function value $M_{ii}$ defined on vertex $i$ and function value $M_{ij}$ defined on edge $(i, j)$. 
Note that because $M$ is a valid matrix, any sublevel set of the resulting function $f: K \to \Rspace$ satisfies the condition that every edge has equal or higher function value than either of its vertices. 
Given a valid matrix, one thus may obtain a labeled merge tree and subsequently an ultra matrix in the following way.

\begin{definition}
\label{def:lmt-valid}
\update{
Let $M \in \R^{n \times n}$ be a valid matrix, $K$ be a  complete graph over $n$ vertices, and $f_M:K \to \R$ be a function induced from $M$ with $f_M(v_i)=M_{ii}$ and $f_M((v_i, v_j))=M_{ij}$. 
The \textbf{labeled merge tree of a valid matrix} $M$, denoted as $\TT(M)$, is the labeled merge tree of the complete graph $K$ with the induced function $f_M$.
} 
\end{definition}

\update{
Basically, given a valid matrix $M$, we can consider $M$ to induce weights of a complete graph $K$ on $n$ vertices. We then compute a minimal spanning tree $\TT(M)$ of this complete graph based on the weights. 
The resulting tree $\TT(M)$ is the labeled merge tree of $M$.
It gives rise to an induced relaxed ultra matrix $\MM\TT(M)$ (recall~\cref{def:relaxedUltrametric}). This procedure corresponds to the maximal subdominant construction in \cite{ChepoiFichet}.}

Note that the labeling is inherited by including internally labeled vertices if there is any pair $i \neq j$ for which $M_{ii} = M_{ij}$.  
See \cref{fig:DegenerateLeaves} for a labeled tree containing an example where $M_{ii} = M_{ij} = M_{jj}$ creates a leaf with two labels, as well as an example where $M_{ii} = M_{ij} > M_{jj}$ creates an internal labeled vertex. 

\begin{lem}
\label{lem:UM_iso_LMT}
$\MM$ induces a \update{bijection} between labeled merge trees and ultra matrices. 
\end{lem}

\begin{proof}
We start with injectivity of $\MM$. 
From \cite[Def.~7.1.2]{Semple2003}, a metric $\delta$ is called a \emph{tree metric} if there exists a weighted $[n]$-tree (i.e., a weighted $X$-tree with $X=[n]$) $(T,f,\pi, \omega)$ for which  $\delta(i,j) = \sum_{e \in \gamma} \omega(e)$ for $\gamma$ the unique path from  $\pi(i)$ to $\pi(j)$ if $\pi(i) \neq \pi(j)$, and is $0$ otherwise. 
By \cite[Thm.~7.1.8]{Semple2003}, such a weighted $[n]$-tree representation is unique. 
For any $(T,f,\pi) \in \LMT$, we can construct a tree metric $\delta_T:[n] \times [n] \to \R_{\geq 0}$ uniquely from $\MM(T,f,\pi)$ by setting
\begin{equation*}
\delta_T(i,j) = 2\MM(T)_{ij} - \MM(T)_{ii} - \MM(T)_{jj}.
\end{equation*} 
To see this is a tree metric, observe that we first assign a length to each edge as the height difference of its two end points. For $i=j$, we are not traveling along any edges, and this formula gives $\delta_T(i,i)=0$, as desired. For $i\neq j$, we want the length of the path from $\pi(i)$ to $\pi(j)$. This path starts at $\pi(i)$ with height $\mathcal{M}(T)_{ii}$, moves upward to the least common ancestor of $\pi(i)$ and $\pi(j)$ at height $\mathcal{M}(T)_{ij}$, and moves downward to $\pi(j)$ at height $\mathcal{M}(T)_{jj}$. The formula combines  these two upward and downward paths  together.

So, given any $(T,f,\pi),(T',f',\pi') \in \LMT$, we construct the two tree metrics $\delta_T$ and $\delta_T'$. However, these two tree metrics are equivalent as $\delta_T'(x,y)\leq \delta_T(x,y)\leq 2\delta_T'(x,y)$. This implies that any continuity condition is the same under either choice of metric.

For ease of notation, denote $\MM(T,f,\pi)$ by $\MM(T)$; similarly for $\MM(T')$. 
If $\MM(T) = \MM(T')$, then $(T,f,\pi,\omega) = (T',f',\pi',\omega')$ as weighted $[n]$-trees by keeping the weighting but ignoring the function values and which vertex is the root.  
Since the function value of any labeled vertex can be determined by $\MM(T)_{ii}$, this implies that $T=T'$ as labeled merge trees. 

Next, we tackle surjectivity of $\MM$.
Given any ultra matrix $M$, we want a labeled merge tree $T$ for which $\MM(T) = M$. 
In particular, we will show that $T = \TT(M)$ satisfies this requirement, which further gives that $\TT$ is the inverse of $\MM$. 
To construct $\TT(M)$, let $K$ be the complete graph on $n$ vertices with vertices labeled $v_1,\cdots,v_n$.  
Define the map $s: K \to \Rspace$ on the complete graph $K$ by $s(v_i) = M_{ii}$ (vertex map) and $s(v_i,v_j) = M_{ij}$ (edge map).
Because $M$ is a valid matrix, this gives a well-defined map; in particular, $s(v_i) \leq s(v_i,v_j)$ for any $i \neq j$.

First, we check that the diagonal entries of the matrices $\MM(\TT(M))$ and $M$ agree. 
By definition of the construction of $\TT(M)$, there is a vertex $\pi(i)$ in the resulting tree with function value $f(\pi(i)) = s(v_i) = M_{ii}$, so clearly $\MM(\TT(M))_{ii} = f(\LCA(\pi(i),\pi(i))) = f(\pi(i)) = M_{ii}$.  

Finally, we check the off-diagonal entries, so assume $i \neq j$ and consider $M_{ij}$. 
Note that $\MM(\TT(M))_{i,j} = f(\LCA(\pi(i),\pi(j)))$ is exactly the function value for which the components containing $v_i$ and $v_j$ merge in the sublevel set  of $s: K \to \Rspace$ (see~\cite{EdelsbrunnerHarer2008} for a discussion on   sublevel set persistence).
Because $s(v_i,v_j) = M_{ij}$, this means that $\MM(\TT(M))_{i,j} \leq M_{ij}$.
Seeking a contradiction, assume that $\MM(\TT(M))_{i,j} < M_{ij}$. 
In order for the components with $v_i$ and $v_j$ to have merged before $M_{ij}$, there must be a path $\gamma = v_iu_1u_2\cdots u_kv_j$ for which every internal edge $e$ has $s(e)<M_{ij}$. 
By the isosceles property using the triangle $v_iv_ju_1$, we know that $s(v_i,v_j) = M_{ij}$ and $s(v_i,u_1)<M_{ij}$, so $s(v_j,u_1) = M_{ij}$. 
The same logic for triangle $u_1u_2v_j$ implies that $s(v_j,u_2) = M_{ij}$.  
Repeating this process for the entire path, we conclude finally that $s(v_j,u_{k-1}) = M_{ij}$. 
However, then the triangle $v_ju_{k-1}u_{k}$ has both $s(u_{k-1},u_{k})$ and $s(u_k,v_j)$ strictly less than $s(u_{k-1},v_j)$, contradicting the isosceles triangle property. 
Thus, we conclude that no such path exists, and therefore $\MM(\TT(M))_{ij} = M_{ij}$.  
\end{proof}

In the course of the above proof, we have showed that $\MM\TT$ is the identity when restricted to ultra matrices, but this is not the case when extending to only valid matrices. 
However, this construction does offer a method for turning a valid matrix into an ultra matrix. 

\begin{definition}
The \textbf{ultra matrix of a valid matrix} $M \in \VM$, denoted $\UU(M)$, is defined to be the induced matrix of $\TT(M)$.
That is, $\UU = \MM\TT$. 
\end{definition}

%-----------------------
\subsection{Available Metrics}
\label{ssec:AvailableMetrics}

There are a number of metrics that may be defined on the space of (labeled) merge trees. 
Note that any metric defined on labeled merge trees can be extended to unlabeled merge trees by simply \update{ignoring} the labeling information, while likely turning the metric into a pseudometric.
In this paper, we focus on interleaving distance $d_I$ and labeled interleaving distance $d_I^L$. 
Other popular distances include the functional distortion distance $d_{FD}$ \cite{Bauer2014} and the bottleneck distance $d_B$.

\paragraph{Interleaving distance}
The interleaving distance is an idea arising from the generalization of the bottleneck distance for persistence diagrams to arbitrary persistence modules \cite{Chazal2009b}. 
Generalizations abound \cite{Bubenik2014a,MunchWang2016,deSilva2018}, but the analog for merge trees was first given in \cite{Morozov2013}. 
We give a modified (non-standard) formulation here, which was shown to be equivalent to the original \cite{Touli2018} \update{(see Theorem 7 of \cite{Touli2018}  for the statement and Appendix A of \cite{Touli2018} for its proof).}
    
\begin{definition}
\label{def:Metric_Interleaving}
Given two merge trees $(T,f),(T',f')$, a \textbf{$\delta$-good} map $\alpha: (T,f) \to (T',f')$ is a continuous map on the metric trees such that the following properties hold:
\begin{itemize}\denselist
\item[(i)] For any $x$ in the geometric realization $|T|$, $f'(\alpha(x)) - f(x) = \delta$;
\item[(ii)] For any $w\in \Image(\alpha)$ with $x' := \LCA(\alpha^{-1}(w))$, $f(x') - f(u) \le 2\delta$ for all $u \in \alpha \inv(w)$; and 
\item[(iii)] For any $w\notin \Image(\alpha)$, $depth(w) \le 2\delta$. 
\end{itemize}
The \textbf{interleaving distance} is then defined to be 
\begin{equation*}
d_I((T,f), (T',f')) = \inf \{\delta \mid \exists\, \delta\text{-good }\alpha:(T,f) \to (T',f') \}.
\end{equation*}
\end{definition}
 
One particularly useful property that we will use later is the following. 
\begin{lem}
\label{lem:GoodMapProperty}
Let $\alpha: (T,f) \to (T',f')$ be a continuous map such that $f'(\alpha(x))  = f(x)  + \delta$  for any $x \in |T|$.
Assume $u \preceq v$.  
Then 
\begin{itemize}
    \item $\alpha(u) \preceq \alpha(v)$, and 
    \item if $w$ is the unique ancestor of $\alpha(u)$ with $f'(w) = f(v) + \delta$, then $w = \alpha(v)$. 
\end{itemize}
\end{lem}
\begin{proof}
Note that $u \preceq v$ implies that $f(u) \leq f(v)$ and further that the unique path $\gamma$ from $u$ to $v$ in $T$ is monotone increasing in $f$. 
Then the image of $\gamma$ in $T'$, $\alpha(\gamma)$, satisfies $f'(\alpha(\gamma(t))) = f(\gamma(t)) + \delta$ and thus is monotone increasing in $f'$. 
Thus, by definition, we have that $\alpha(u) \preceq \alpha(v)$. 
Further, the uniqueness of paths implies that if $w$ is the unique ancestor with $f'(w) = f(v) + \delta$, then it must be the endpoint of $\gamma$, and so $w = \alpha(v)$. 
\end{proof}

\paragraph{Labeled interleaving distance}
The following metric is closely related to one originally defined in \cite{Cardona2013} for comparing phylogenetic trees.

\begin{definition}
\label{def:Metric_LabeledInterleaving}
\update{Given two labeled merge trees sharing the same set of $n$ labels}, the \textbf{labeled interleaving distance} is 
\begin{equation*}
d_I^L((T,f,\pi), (T',f',\pi')) = \|\MM(T,f,\pi) -  \MM(T',f',\pi') \|_\infty.
\end{equation*}
\end{definition}

\noindent The reason for calling such a distance an interleaving distance comes from \cite{Munch2018} where it is shown that this metric arises as an interleaving distance on a particular category with a flow \cite{deSilva2018}. 
Note that because we need the labels in order to be able to have a well-defined matrix, this metric only works on labeled merge trees. 

%-----------------------
\subsection{Intrinsic Metrics}
Given a metric $d$ on merge trees, we may define its intrinsic version as follows; see, e.g., \cite{Burago2001}. 

\begin{definition}
\label{def:Metric_IntrinsicBottleneck}
Given two merge trees, let $\gamma:[0,1]\rightarrow \MT$ be a continuous path in $d$ 
such that $\gamma(0)=T$ and $\gamma(1)=T'$. 
The \textbf{length of $\gamma$ induced by the distance $d$} is defined as 
\begin{equation*}
L_{d}(\gamma)=\ds \sup_{n,\sum}\;\sum_{i=0}^{n-1} d(\gamma(t_i),\gamma(t_{i+1})),
\end{equation*}
where $n$ ranges over $\mathbb{N}$ and $\sum$ ranges over all partitions $0=t_0\leq t_1\leq \ldots \leq t_n=1$ of $[0,1]$. 
The \textbf{intrinsic metric $\hat{d}$ induced by the distance $d$} is 
\begin{equation*}
\hat{d}(T,T')=\ds\inf_\gamma L_{d}(\gamma).
\end{equation*}
\end{definition}

\noindent Thus, the induced intrinsic metric on a metric space is the infimum of the lengths of all paths from one point to another. It is known that $d$ is always less than or equal to $\hat{d}$. 

A metric space is said to be a \emph{length space} if the original metric $d$ coincides with the intrinsic metric $\hat{d}$. Recall that a metric space is said to be a \emph{geodesic space} if any two points in the space can be connected by a curve of length equal to the distance between
the two points.
In this case, the metric is said to be \emph{strictly intrinsic}.
 Note that a geodesic space is necessarily a length space.

%-------------------------
% 1-center

\section{Geodesics and 1-Centers for Labeled Merge Trees}
\label{sec:labeled}

In this section, we prove an inequality involving the labeled interleaving distance and provide methods for constructing geodesics and 1-centers for collections of labeled merge trees. 

\subsection{More on the Labeled Interleaving Distance}

\update{ 
The following result is not new. It follows from Theorem 2 of \cite{SCM16}, which is slightly more general than the lemma below in the sense that the matrices $M$ and $M'$ are allowed to be non-valid as well. It is also a slight generalization of Lemma 15 of \cite{CM10} (which is the following result restricted to the metric setting). It can also be deduced from  Proposition~1 and Corollary~1 of \cite{ChepoiFichet} because the set of relaxed ultra matrices (which are the matrices corresponding to merge trees) is stable under translations along the diagonal. See Section~3 of \cite{ChepoiFichet} for the special case of ultra matrices (equivalently, labeled dendrograms), which adapts straightforwardly to our slightly more general setting. Neverthelss, here we provide simple and direct proofs, both for completeness and clarity. 
}

\begin{lem}
\label{lem:ultrammt-Lip}
For any pair of valid matrices $M, M' \in \VM$, 
\begin{equation*}
d_I^L(\TT(M),\, \TT(M')) \leq \|M-M'\|_\infty.
\end{equation*}
\end{lem}
\begin{proof}
Since, by definition, 
$d_I^L(\TT(M),\, \TT(M')) = \|\UU(M)-\UU(M')\|_\infty$, 
we will actually establish the inequality $\|\UU(M)-\UU(M')\|_\infty \leq \|M-M'\|_\infty$.

Let $\delta = \|M-M'\|_\infty$. 
Let $T=\MM(M)$ and $T'=\MM(M')$ be the associated merge trees, and $\widetilde M = \UU(M)$ and $\widetilde M' = \UU(M')$ the induced ultra matrices. 
Consider any pair of (possibly equal) labels $i$ and $j$ with $1\leq i\leq j\leq n$.
We consider the vertices $v_i$ and $v_j$ in the complete graph $K$ with $s, s':K \to \R$ the maps on $K$ induced by $M$ and $M'$, respectively. 
As $v_i$ and $v_j$ are in the same component of the $(M_{ij})$-sublevel set of $s$ (i.e.,~sublevel set of $s$ at value $M_{ij}$, $s\inv(\infty, M_{ij}]$), there is a path $\gamma$ in $K$ with $s(e) \leq \widetilde M_{ij}$ for all edges $e$ in the path. 
Because $\|M-M'\|_\infty\leq \delta$, we have that 
\begin{equation*}
    s'(e) \leq s(e) + \delta \leq \widetilde{M_{ij}} + \delta
\end{equation*}
for every $e \in \gamma$. 
So, $v_i$ and $v_j$ are in the same component of the $(\widetilde M_{ij} + \delta)$-sublevel set of $s'$ and thus $\widetilde{M'}_{ij} \leq \widetilde M_{ij} + \delta$. 

Symmetrically, for any $t<\widetilde M_{ij}-\delta$, $v_i$ and $v_j$ do not lie in the same connected component of the $t$-sublevel set of $s'$. 
Otherwise, by the same argument as above, $v_i$ and $v_j$ would belong to the same connected component of the $(t+\delta)$-sublevel set of $T$ with $t+\delta < \widetilde M_{ij}$, a contradiction. 
Hence, $\widetilde M'_{ij} \geq \widetilde M_{ij} - \delta$. It follows that $|\widetilde M'_{ij} - \widetilde M_{ij}|\leq \delta$, and since this is true for all labels $1\leq i\leq j\leq n$, the symmetric matrices $\widetilde M, \widetilde M'$ satisfy $\|\widetilde M-\widetilde M'\|_\infty \leq \delta$. 
Hence, $d_I^L(T, T') = \|\widetilde M-\widetilde M'\|_\infty \leq \delta$.
\end{proof}

\subsection{Geodesics in $\LMT$}

The next corollary looks at the straight line between the matrices associated to two labeled merge trees. 
Specifically, given any two labeled merge trees $T, T' \in \LMT$, we know that their associated matrices $M = \MM(T), M' = \MM(T')$ are ultra matrices. 
We can define the line between them by setting $M^\lambda := (1-\lambda)M + \lambda M'$ for $\lambda \in [0,1]$. 
While not necessarily ultra matrices, it is easy to check that $M^\lambda \in \VM$ for all $\lambda \in [0,1]$. 
We can then pull this back to a path of labeled merge trees by setting $T^\lambda = \TT(M^\lambda)$.

\begin{cor}[LMT Geodesics]
\label{cor:labelintrinsic}
Given any two labeled merge trees $T, T' \in \LMT$, and their corresponding 
ultra matrices $M = \MM(T), M' = \MM(T')$, the family of merge trees $\left\{T^\lambda:=\TT\left(M^\lambda\right)\right\}_{\lambda\in[0,1]}$ defines a geodesic between $T$ and $T'$ in the metric $d_I^L$. 
As a consequence, on the space of labeled merge trees, the metric $d_I^L$ is strictly intrinsic. 
\end{cor}
\begin{proof}
Let $\delta$ denote the distance $d_I^L(T, T') = \|M-M'\|_\infty$. 
For any $0\leq \lambda\leq \lambda'\leq 1$, the linearly interpolating matrices $M^\lambda, M^{\lambda'}$ satisfy $\|M^\lambda-M^{\lambda'}\|_\infty \leq (\lambda'-\lambda)\,\delta$. 
Hence, by \cref{lem:ultrammt-Lip}, we have $d_I^L(T^\lambda, T^{\lambda'}) \leq (\lambda'-\lambda)\,\delta$. 
Since this is true for all $0\leq\lambda\leq\lambda'\leq 1$, the triangle inequality implies that the family $\{T^\lambda\}_{\lambda\in[0,1]}$ defines a geodesic between $T$ and $T'$.
\end{proof}

See the example of \cref{fig:AverageMergeTreeExample}. 
Setting $\lambda = 1/2$, $M^\lambda$ is the matrix (labeled $M$) shown in the middle green circle, and $T^\lambda$ (labeled $\TT(M)$) is the tree shown at the far right. 
\update{\cref{cor:labelintrinsic} discusses the geodesics in the space of labeled merge trees. A metric space in general may have no geodesics; thus~\cref{cor:labelintrinsic} provides an additional property for the space of interest. Furthermore, a geodesic can be used to perform shape morphing between a pair of merge trees (see~\cite{YanWangMunch2020})}. 

\begin{figure}[!ht]
\centering
\includegraphics[width = .5\textwidth]{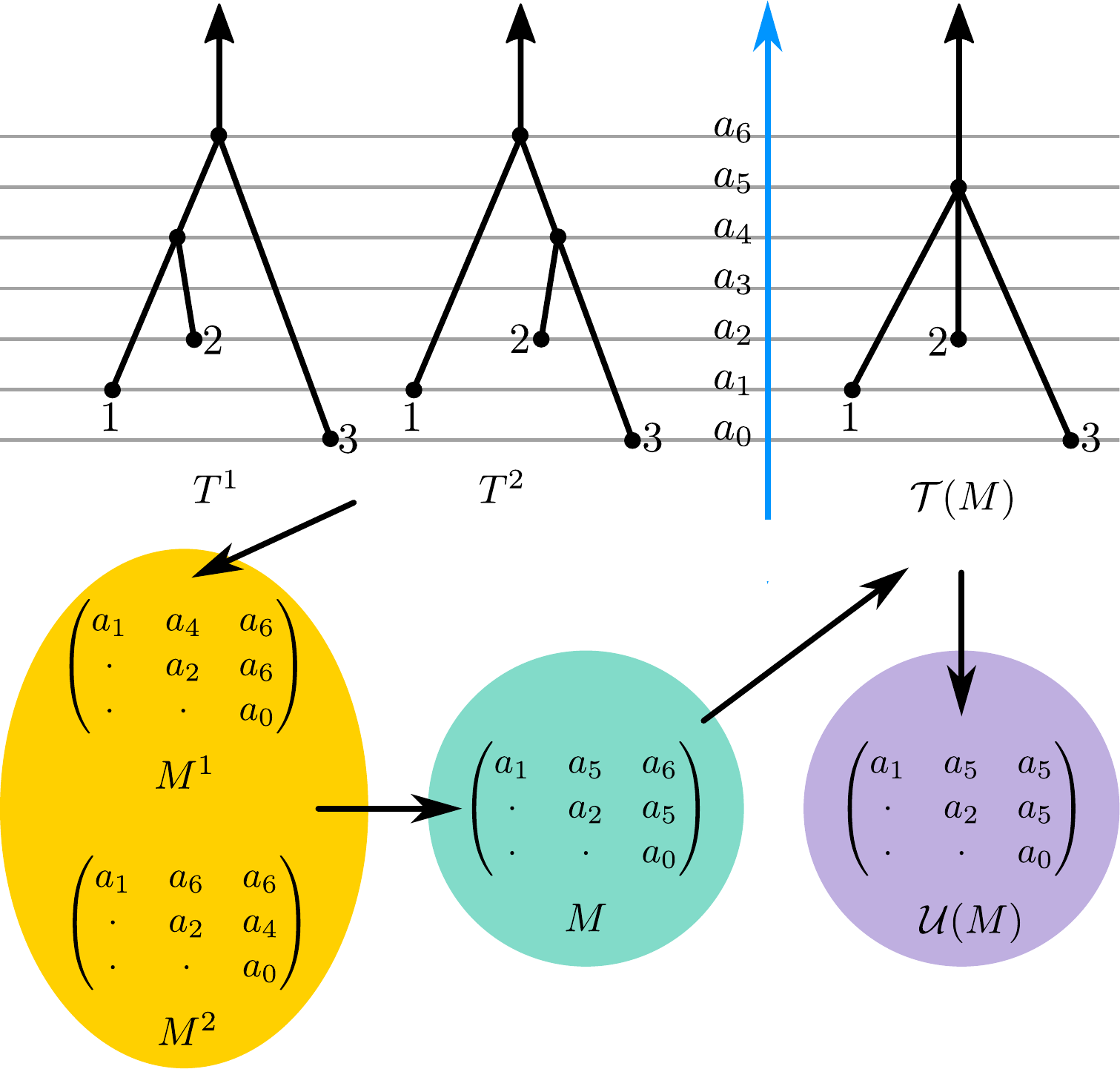}
\caption{An example of the averaging process for labeled merge trees. 
$T^1$ and $T^2$ are labeled merge trees with induced matrices $M^1$ and $M^2$.
$M$ is the pointwise average of $M^1$ and $M^2$, but is not an ultra matrix. 
The labeled merge tree $\TT(M)$ is shown, whose induced matrix is the ultra matrix $\UU(M)$.}
\label{fig:AverageMergeTreeExample}
\end{figure}

\subsection{1-centers in $\LMT$}
\label{subsec:1centers}

\update{Our $1$-center merge tree originates from the notion of a metric $k$-center in graph theory. 
Given $m$ number of cities, one aims to build $k$ facilities that minimize the maximum distance between a city to a facility.}
For $k=1$, a metric 1-center of a finite set of labeled merge trees is one that minimizes the maximum distance to any other tree in the set.  
A metric 1-center may or may not be unique. 
\begin{definition}
Given a metric space $(X, d)$, a \textbf{$1$-center} $c \in X$ of a finite point set $P = \{p_1, \cdots, p_m\} \subset X$ is 
\begin{equation*}
c\; \update{ \in} \argmin_{x \in X} \max_{p \in P} d(x, p). 
\end{equation*}
That is, $c$ is a center of the minimum enclosing ball of $P$.
\end{definition}

\update{
Here, we use the set notation $\in$ to indicate that $c$ may not be unique. 
In the case of a finite collection of numbers $\chi$ in $\R$, the 1-center is simply the midpoint of the enclosing interval, $(\max(\chi) + \min(\chi))/2$.
Now suppose we are given a collection of matrices $\{M^1,\cdots,M^N\}$. 
Let $M_{mid}$ denote the matrix consisting of the entry-wise 1-center of the matrices, i.e., $M[i][j]$ is the midpoint of the enclosing interval of numbers $\{M^1[i][j], M^2[i][j], ..., M^N[i][j]\}$. It is easy to see that $\Mmid$ is a 1-center for these matrices in the space of all matrices equipped with the $\ell^\infty$ norm. 
A similar statement holds for a collection of valid matrices, and we include its simple proof for completeness.  
\begin{claim}
Let $M^1,\cdots,M^N$ be valid $n\times n$ matrices, and $\Mmid$ be the matrix consisting of the entry-wise 1-center of these matrices. Then $\Mmid$ must be valid as well and $\Mmid$ is a 1-center of $\{M^1, \ldots, M^N\}$ in the space of valid matrices equipped with the $\ell^\infty$ norm. 
\label{claim:validMat1center}
\end{claim}
}

\begin{proof}
\update{
In what follows, all spaces of matrices are equipped with the $\ell^\infty$ norm. 
Since the space of valid matrices is a subspace of the space of all matrices, it follows that $\Mmid$ is a 1-center of $\{M^1, \ldots, M^N\}$ in the space of all matrices. Hence to prove the claim we only need to show that $\Mmid$ is a valid matrix. In other words, $\Mmid[i][i] \le \Mmid[i][j]$ for any $i,j \in [n]$. To see why this holds, note that for any $i,j \in [n]$,
$$
\Mmid[i][i] = \frac{\max_k(M_{ii}^k) + \min_k(M_{ii}^k)}{2}\leq \frac{\max_k(M_{ij}^k) + \min_k(M_{ij}^k)}{2} = \Mmid[i][j].$$
The claim thus follows.}
\end{proof}

\update{
$\Mmid$ as a 1-center of valid matrices is, by itself, a valid matrix, but may not be an ultra matrix, so we can replace it by its labeled merge tree (following the procedure described by~\cref{def:lmt-valid}) and take its corresponding ultra matrix, thus turning it back to an ultra matrix.} 

\update{
The main result of this section is an algorithm to compute the 1-center of a collection of labeled merge trees under the labeled interleaving distance $d_I^L$. In particular, suppose we are given a set of labeled merge trees $\{T^1, \ldots, T^N\}$, whose corresponding induced matrices $\{\M^1,\cdots,\M^N\}$ are both valid and ultra. 
We compute a 1-center valid matrix $\Mmid$ of $\{\M^1,\cdots,\M^N\}$ following~\cref{claim:validMat1center}, and convert it to a labeled merge tree, denoted $T^*$. 
Then $T^*$ is a 1-center of the labeled merge trees, see \cref{fig:AverageMergeTreeExample} for a simple example. 
The correctness of this procedure is established in the following~\cref{prop:1center_v2}. 
}

\begin{prop}[LMT 1-Center]
\label{prop:1center_v2}
Let $\{T^1,\cdots,T^N\}$ be a set of labeled merge trees, which gives rise to a set of valid and ultra matrices $\{\M^1,\cdots,\M^N\}$. 
Let $T^*$ be a merge tree constructed as above. Then $T^*$ is a 1-center of $\{T^1,\ldots, T^N\}$.
Furthermore, let $U^* = \MM(T^*) = \MM \TT(\Mmid)$ be the ultra matrix corresponding to $T^*$. Then $U^*$ is a 1-center of the set of ultra matrices $\{M^1,\ldots, M^N\}$.  
\end{prop}

\begin{proof}
\update{Recall that (the valid matrix) $\Mmid$ is the 1-center of ultra matrices $\{M^1,\ldots,M^N\}$ in the space of valid matrices following~\cref{claim:validMat1center}. Set $\delta = \ds\max_{i} \|\Mmid - M^i\|_\infty$. 
Then $d_I^L(T,T^i) \leq \|\Mmid- M^i\|_\infty$ by \cref{lem:ultrammt-Lip}. It then follows that
\begin{equation*}
    \max_i d_I^L(T^*,T^i) \leq \max_i \|\Mmid - M^i\|_\infty \leq \delta.
\end{equation*}
Thus $\{T_i\}_{i=1}^N$ is contained in a ball of radius $\delta$ centered at $T^*$. 
}

\update{We now show that this is in fact a \emph{minimum enclosing ball} of $\{T^1,\ldots,T^N\}$ in the space of labeled merge trees, which would then imply that $T^*$ is a 1-center for these merge trees.
Specifically, assume there exists a $\widetilde T$ such that $\ds\max_i d_I^L(\widetilde T, T^i) < \delta$. 
Set $\widetilde U = \MM(\widetilde T)$. 
Then for any $i$, 
\begin{equation*}
    \|\widetilde U - M^i\|_\infty 
    = d_I^L(\TT(\widetilde U), \TT(M^i)) 
    = d_I^L (\widetilde T, T^i) < \delta.
\end{equation*}
Hence $\widetilde U$, as a valid matrix, gives rise to a smaller $\max_i \|\widetilde U - M^i\|$, which contradicts the assumption that $\Mmid$ is a 1-center within the space of valid matrices (i.e, $\Mmid = \argmin_M \max_i \| M - M^i\|$). Hence such a $\widetilde{T}$ cannot exist, and $T^*$ is a 1-center for $\{T^1, \cdots, T^N\}$. By the relation between distance for ultra matrices and for their corresponding labeled merge trees, $U^*=\MM(T^*)$ is a 1-center for $\{M^1, \cdots, M^N\}$, as well. 
}
\end{proof}
\update{\paragraph{Remark.} As a corollary of the above result, if we are given a collection of ultra matrices $\{M^1, \ldots, M^N\}$, then $U^* = \MM \TT (\Mmid)$ is a 1-center for them in the space of ultra matrices, where $\Mmid$ as defined earlier is the matrix consisting of the entry-wise 1-center of the input ultra matrices and $\Mmid$ is itself not necessarily a ultra matrix. 
Computing 1-centers for ultrametrics has been  explored in the literature. While in general, this problem is NP-hard, for the case when we consider the $\ell^\infty$-norm on the space of ultrametrics (which is the same as our setting), it is known that there is a simple algorithm to compute it \cite{ChepoiFichet}. However, our approach above is completely different from the previous approach in \cite{ChepoiFichet}, and has a different interpretation as well. 
}

%-------------------------
% unlabeled MT

\section{Interleaving Distances for Unlabeled Merge Trees}
\label{sec:unlabeled}

Moving to the unlabeled setting, we establish the existence of a certain labeling for a pair of merge trees that allows us to show that the interleaving distance for unlabeled merge trees is intrinsic.

\begin{thm}
\label{thm:LabeledVsUnlabeledInterleaving}
Given two merge trees $(T,f)$ and $(T',f')$, let $L$ and $L'$ be the respective leaf sets.
Then 
\begin{equation}
\label{eq:LabeledVsUnlabeledInterleaving}
d_I((T,f),(T',f')) = \inf_{\pi,\pi'} d_I^L((T, f,\pi) ,(T', f',\pi'))
\end{equation}
where the infimum is taken over all finite labelings of the two given merge trees, $\pi$ and $\pi'$,  using at most $|L| + |L'|$ labels. 
\end{thm}

Prior to proving the theorem, we will investigate the following construction of a labeling when given a $\delta$-good map. 
First, note that given two labeled merge trees $(T, f,\pi)$ and $(T', f',\pi')$, where $\pi: [n] \to V(T)$ and $\pi':[n]\to V(T')$, the labeling information can be equivalently stored as an ordered collection of pairs $\Pi = \{ (\pi(i), \pi'(i))  \mid i \in [n] \} \subseteq V(T) \times V(T')$. 
Since the order of the labels does not matter for this particular application, we will build $\Pi$ iteratively and assign the integers at the end. 

Let $L$ and $L'$ denote the leaf sets for $T$ and $T'$, respectively. 
Assume we are given a $\delta$-good map $\alpha$ as described in \cref{def:Metric_Interleaving}.
While this map is defined on the underlying metric trees, note that we can subdivide the trees so that $\alpha(v)$ is a vertex in $T'$ for any vertex in $T$, and further that every point in the set $\alpha\inv(w)$ is a vertex in $T$ if $w$ is a vertex in $T'$.

Then, we construct the labeling $\Pi$ as follows. 
\begin{description}\denselist
\item[(S-1)] Fix some $v \in L$, and let $w = \alpha(v)$. 
Then for every $u \in \alpha \inv(w)$, add $(u,w)$ to $\Pi$. 
Repeat this for every vertex in $L$. 
\item[(S-2)] For any leaf node $w\in L' \setminus \Image(\alpha)$, let $x$ be its lowest ancestor contained in  $\Image(\alpha)$. 
Let $u \in \alpha^{-1}(x)$ be an arbitrary preimage of $x$ from $|T|$. 
Add $(u,w)$ to $\Pi$. 
Repeat for all leaves  in $L'$. 
\item [(S-3)] Fix an ordering on the pairs in $\Pi = \{(u_i,w_i) \mid i \in [n] \}$ and define $\pi(i) = u_i \in T$ and $\pi'(i) = w_i \in T'$. 
\end{description}

Observe that since the preimage of any leaf node $w\in L' \cap \Image(\alpha)$ must be some vertex (or vertices) in $L$, any $w\in L' \cap \Image(\alpha)$ will be paired with some $u\in L$ by the process in (S-1), so this procedure does not miss any leaves in $T'$.
See \cref{fig:LabelsInducedByMap} for an example.

\begin{figure}[!ht]
    \centering
    \begin{minipage}{.34\textwidth}
    \centering 
    $\MM(T,f,\pi) = $ \\
    	$
    	\begin{pmatrix}
	a_1   & a_4   & a_7   & a_7   & a_1   & a_4  & a_7 \\
	\cdot & a_2   & a_7   & a_7   & a_4   & a_2  & a_7\\
	\cdot & \cdot & a_4   & a_5   & a_7   & a_7  & {\color{red}a_5}\\
	\cdot & \cdot & \cdot & a_0   & a_7   & a_7  & {\color{red}a_4}\\
	\cdot & \cdot & \cdot & \cdot & a_1   & a_4  & a_7\\
	\cdot & \cdot & \cdot & \cdot & \cdot & a_2  & a_7\\
	\cdot & \cdot & \cdot & \cdot & \cdot & \cdot & a_4\\
	\end{pmatrix}$
    \end{minipage}
    \begin{minipage}{.3\textwidth}
    \includegraphics[width = \textwidth] {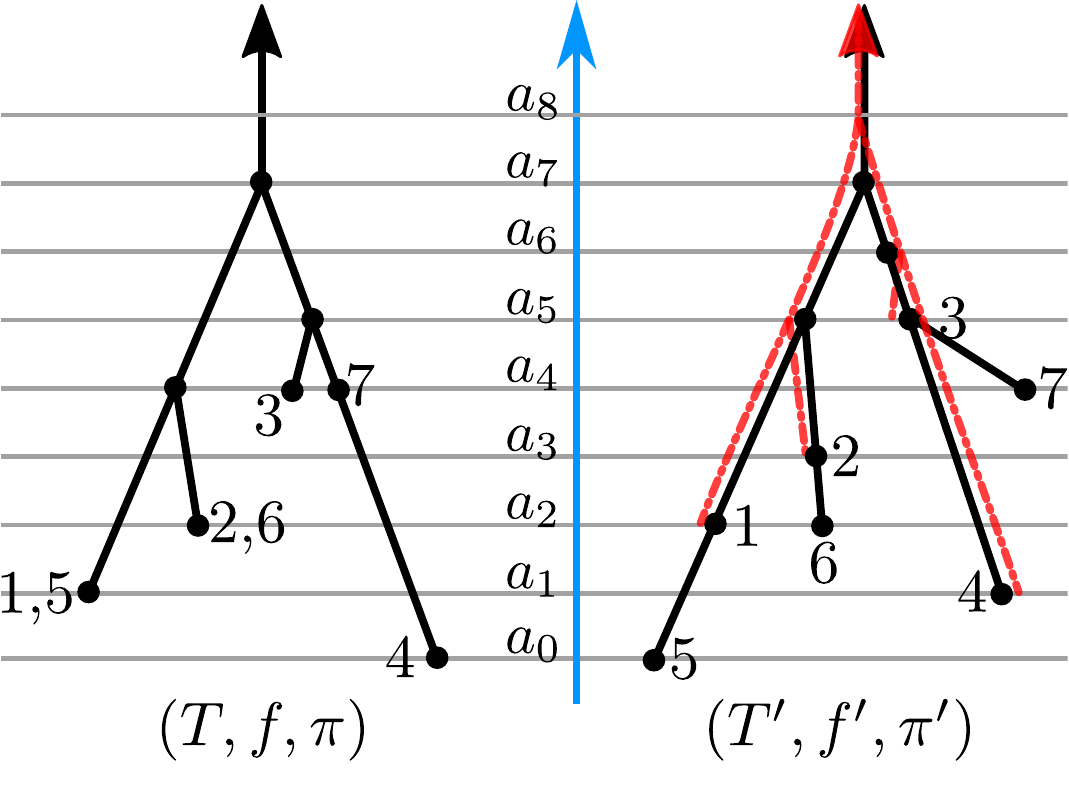}
    \end{minipage}
        \begin{minipage}{.34\textwidth}
    \centering 
    $\MM(T',f',\pi') = $ \\
    	$
    	\begin{pmatrix}
	a_2   & a_5   & a_7   & a_7   & a_2   & a_5  & a_7 \\
	\cdot & a_3   & a_7   & a_7   & a_5   & a_3  & a_7\\
	\cdot & \cdot & a_5   & a_5   & a_7   & a_7  & a_5\\
	\cdot & \cdot & \cdot & a_1   & a_7   & a_7  & a_5\\
	\cdot & \cdot & \cdot & \cdot & a_0   & a_5  & a_7\\
	\cdot & \cdot & \cdot & \cdot & \cdot & a_2  & a_7\\
	\cdot & \cdot & \cdot & \cdot & \cdot & \cdot & a_4\\
	\end{pmatrix}$
    \end{minipage}
    \caption{Given $\alpha: (T,f) \to (T',f')$, this is an example of the labeling induced by the procedure discussed after \cref{thm:LabeledVsUnlabeledInterleaving}.
    The image of the map $\alpha$ is given by the red dashed lines, and $\alpha$ is $\delta$-good for $\delta = a_{i+1}-a_i$.   
    Labels 1-4 were generated in (S-1), the rest in (S-2). 
    Note that there were two options for the location of label 7 in $T$.
    The other choice would be the same as the vertex labeled 3, and would only change the red entries in $\MM(T,f,\pi)$.}
    \label{fig:LabelsInducedByMap}
\end{figure}

To use this construction to prove \cref{thm:LabeledVsUnlabeledInterleaving}, we will use the following two lemmas. 

\begin{lem}
\label{lem:funcperturb}
For any $(u,w) \in \Pi$, $|f(u) - f'(w)| \leq \delta$. 
\end{lem}
\begin{proof}
If $(u,w)$ is generated from (S-1) above, then the lemma holds by property (i) in the definition of the $\delta$-good map $\alpha$ (see~\cref{def:Metric_Interleaving}). 
If $(u,w)$ is generated from (S-2), then the lemma follows from property (iii) of the $\delta$-good map $\alpha$. 
Indeed, let $x$ be the lowest ancestor of $w$ contained in  $\Image(\alpha)$, so that $\alpha(u) = x$.  
Then $0 \le f'(x) - f'(w) \le 2\delta$ and $f'(x) - f(u) =\delta$, implying that 
$|f'(w) - f(u)| \le \delta$. 
\end{proof}

\begin{lem}
\label{lem:distperturb}
For any $(u_1, w_1), (u_2, w_2) \in \Pi$, 
$|f(\LCA(u_1, u_2)) - f'(\LCA(w_1, w_2))| \le \delta$. 
\end{lem}

\begin{proof}
Assume we are given $\alpha$, a $\delta$-good map. 
If $(u_i,w_i)$ is generated from (S-1), set $w_i' = w_i$. 
If $(u_i, w_i)$ is generated via (S-2), then let $w_i'$ be the lowest ancestor of $w_i$ in $\Image(\alpha)$. 
In both cases, we have that $\alpha(u_i) = w_i'$ and $w_i \preceq w_i'$. 

Set $u_0 = \LCA(u_1, u_2)$, $w_0 = \LCA(w_1, w_2)$  and $w_0' = \LCA(w_1', w_2')$. 
We will first show that $w_0 = w_0'$.
If both pairs come from (S-1), then  $w_i = w_i'$ and the claim is obvious. 
So, assume that at least one, say $(u_1,w_1)$, comes from (S-2) and thus $w_1 \neq w_1'$. 
As $w_i \preceq w_i' \preceq w_0'$ for each $i$, the least common ancestor property implies $w_0 \preceq w_0'$. 
Seeking a contradiction, assume that $w_0$ is not a common ancestor of both $w_i'$; without loss of generality, say $w_0$ is not an ancestor of $w_1'$. 
Let $z = \LCA(w_0,w_1',w_2')$. 
Then there are two paths in $T'$ from $w_1$ to $z$: one through $w_0$ and one through $w_1'$. 
This contradicts the tree assumption of $T'$.
Therefore, $w_0$ is a common ancestor of $w_i'$, implying $w_0' \preceq w_0$, and so $w_0 = w_0'$.

We will now prove the main claim, namely, that $|f(u_0) - f'(w_0)| \le \delta$.
To see that this is the case, assume that the claim does not hold; that is, either $f(u_0)-f'(w_0) > \delta$ or $f'(w_0) - f(u_0) > \delta$. 
Suppose first that $f'(w_0) - f(u_0) > \delta$, and consider $\alpha(u_0)$. 
Because $u_i \preceq u_0$ for $i=1,2$, by \cref{lem:GoodMapProperty} we must have that $w_i' = \alpha(u_i) \preceq \alpha(u_0)$ for $i=1,2$. 
However, then $\alpha(u_0)$ is an ancestor of both $w_1'$ and $w_2'$ with 
\begin{equation*}
    f'(\alpha(u_0)) = f(u_0) + \delta < f'(w_0),
\end{equation*}
contradicting the least common ancestor assumption of $w_0$.

Next, suppose $f(u_0) - f'(w_0) > \delta$ and 
consider $\alpha^{-1}(w_0)$. 
We claim that any point in $\alpha^{-1}(w_0)$ is a descendant of $u_0$; i.e.,~$v \preceq u_0$ for all $v \in \alpha\inv(w_0)$. 
Otherwise, we have that 
\begin{equation*}
    f(\LCA(\alpha^{-1}(w_0))) > f(u_0) > f'(w_0) + \delta = f(v) + 2\delta
\end{equation*}
for any $v \in \alpha \inv(w_0)$, contradicting property (ii) of \cref{def:Metric_Interleaving}. 
For $i=1,2$, let $v_i$ be the unique ancestor of $u_i$ with $f(v_i) = f'(w_0) - \delta$. 
By \cref{lem:GoodMapProperty}, since $\alpha(u_i) = w_i'$ and $w_0$ is the unique ancestor of $w_i'$ with $f'(w_0) = f(v_i) + \delta$, this implies that $\alpha(v_i) = w_0$.
That is, $v_i \in \alpha \inv(w_0)$. 
Further, $v_1 \neq v_2$. 
Otherwise if $v := v_1 = v_2$, then 
\begin{equation*}
    f(v) = f'(w_0) - \delta < f(u_0) - 2\delta < f(u_0)
\end{equation*}
and thus $v$ is a lower common ancestor of $u_1$ and $u_2$ than $u_0$, a contradiction. 
Hence, $\LCA(v_1,v_2) = u_0$. 
However, 
\begin{equation*}
    f(u_0) - f(v_i) 
    = f(u_0) - f'(w_0) + \delta > 2\delta.
\end{equation*}
This also contradicts property (ii) of \cref{def:Metric_Interleaving}, finishing the proof of \cref{lem:distperturb}.
\end{proof}

\begin{proof}[Proof of \cref{thm:LabeledVsUnlabeledInterleaving}]
Say we have a $\delta$-good map $\alpha$ for some $\delta \geq d_I((T,f),(T',f'))$. 
We construct the labelings $\pi,\pi'$ as described above. Then \cref{lem:funcperturb} and \cref{lem:distperturb} imply that \[d_I^L((T,f,\pi),(T',f',\pi')) \leq \delta.\] 
As this is true for any $\delta$, $\ds\inf_\Pi d_I^L((T,f,\pi),(T',f',\pi')) \leq d_I((T,f),(T',f'))$.

To show the other inequality, assume we are given any pair of labelings $\pi$, $\pi'$ and assume 
\begin{equation*}
d_I^L((T,f,\pi),(T',f',\pi')) = \delta.
\end{equation*}
We will construct the map $\alpha$ and show that it is $\delta$-good. 
For any $x \in |T|$, let $S_x \subseteq [n]$ be the labels in the subtree of $x$. 
Let $y_i$ be the unique ancestor of $\pi'(i) \in |T'|$ for $i \in S_x$ with $f'(y_i) = f(x) + \delta$. 
First, we note that $y_i = y_j$ for all $i,j \in S_x$. 
Indeed, let $M = \MM(T,f,\pi)$ and $M' = \MM(T',f',\pi')$. 
Then we know $M_{ij}' \leq \delta + M_{ij}$ and so 
\begin{equation*}
    f'(y_i) = f(x) + \delta \geq f(\LCA(\pi(S_x))) =  M_{ij}+ \delta \geq M_{ij}' = f'(\LCA(\pi'(S_x))).
\end{equation*}
Because every $y_i$ has function value greater than the lowest common ancestor of $\pi'(S_x)$, the tree property implies that all $y_i$ are equal. 
Thus, we can set $\alpha(x) = y_i$ for any $i \in S_x$ and it is well-defined. 

We need to ensure that the $\alpha$ constructed is $\delta$-good as given in \cref{def:Metric_Interleaving}.
The map satisfies property (i) by construction, so we move on to (ii). 
Let $w \in |T'| \cap \Image (\alpha)$ and set $x' = \LCA(\alpha\inv(w)) \in |T|$. 
Fix any $u \in \alpha \inv(w)$, and clearly $f(u) \leq f(x')$. 
Now $x'$ must be $\LCA(u,u')$ for some other $u' \in \alpha\inv(w)$. 
Let $i$ be a label in the subtree of $u$, and let $j$ be a label in the subtree of $u'$. 
This further implies that $x' = \LCA ( \pi(i), \pi(j))$. 
Set $w' = \LCA(\pi'(i),\pi'(j))$ and note that as $\pi'(i) \preceq w$ and $\pi'(j) \preceq w$, this implies that $w' \preceq w$. 
In particular, this means $f'(w') \leq f'(w)$.
Further, by assumption $ |f(x') - f'(w') | = |M_{ij} - M'_{ij}| \leq \delta$. 
Thus, 
\begin{equation*}
    f(x') - f(u) \leq 
    (f'(w) - f(u)) + (f(x') - f'(w')) + (f'(w') - f'(w)) \leq 2\delta
\end{equation*}
as the first part of the middle term is exactly $\delta$, the second is $\leq \delta$, and the last is negative, showing that $\alpha$ satisfies property (ii). 

Finally, we ensure property (iii). 
Let $w \in |T'| \setminus \Image(\alpha)$. 
Let $i$ be the label of any leaf in the subtree of $w$, and set $y = \alpha(\pi(i))$ to be the image of the vertex labeled $i$ in $T$. 
Then the tree property implies that $\pi'(i) \preceq w\preceq y$ and thus $f'(\pi'(i)) \leq f'(w) \leq f'(y)$.  
So, 
\begin{equation*}
    |f'(w) - f'(\pi'(i))|  \leq |f'(\pi'(i)) - f'(y) |
                            \leq |f'(\pi'(i)) - f(\pi(i))| +  \delta 
                            = |M_{ii} - M_{ii}'| +  \delta \leq 2\delta.
\end{equation*}
As this is true for every leaf in the subtree of $w$, $\depth(w) \leq 2\delta$ and so $\alpha$ satisfies property (iii).

Thus, we have that $ d_I((T,f),(T',f')) \leq d_I^L((T,f,\pi),(T',f',\pi'))$ for any given $\Pi$, completing the proof of the theorem.
\end{proof}

We can use the construction from the proof to state something stronger. 
Recall that we work with finite labeled and unlabeled merge trees throughout the paper. 

\begin{cor}
\label{cor:OptLabel}
There exist an $n$ and a pair of labelings $\pi,\pi'$ so that 
\begin{equation*}
    d_I((T,f), (T',f')) = d_I^L((T,f,\pi), (T',f',\pi')) ,
\end{equation*}
where $L$ (resp.~$L'$) is the set of leaves of $T$ (resp.~$T'$). Thus, the interleaving distance for finite merge trees is always achieved by a map $\alpha$.
\end{cor}
\begin{proof}
The right side of Equation \eqref{eq:LabeledVsUnlabeledInterleaving} in~\cref{thm:LabeledVsUnlabeledInterleaving} is taken over labelings using at most $N = |L| + |L'|$ labels, which is finite. (Here, $L'$ is the set of leaves of 
$T'$.)
Up to reordering, we can use the first $|L|$ numbers to label the leaves in $T$ and the last $|L'|$ to label the leaves in $T'$. 
All that remains to show is that there are finitely many possible locations to place the remaining labels in each tree. 
Indeed, if $d_I(T,T') = \delta$, then for each $i \in \{1,\cdots,|L|\}$, one has the option of placing $i$ at any point in  $(f')\inv(f(\pi(i)) + \delta) \subset T'$. 
Note that $|(f')\inv(f(\pi(i)) + \delta)| $ is finite.
Similarly, there are $|f\inv(f'(\pi'(i)) + \delta)|$ possible locations available for $i \in \{|L|+1, N\}$ to be placed in $T$.
For any fixed choice from this set for every $i$, let $M$ and $M'$ be the associated matrices for $T$ and $T'$, respectively.

The options are set up so that any choice of location for label $i$ in the opposite tree will automatically satisfy $|M_{ii} - M_{ii}'| = \delta$, so we need only ensure that some choice in each tree of these locations for every $i$ promises $|M_{ij} - M_{ij}'| \leq \delta$. 
For every choice of remaining labels, say there is some $i,j$ for which $|M_{ij} - M_{ij}'| > \delta$.
As we have finitely many options, there is an $\epsilon$ so that $|M_{ij} - M_{ij}'| > \delta + \epsilon$. 
However, there is certainly a $(\delta + \epsilon/2)$-good map $\alpha$ that does not take the labels into consideration, and we could then build the labeling as discussed in \cref{thm:LabeledVsUnlabeledInterleaving}, giving a contradiction. 
Thus, one of finitely many options achieves the left infimum of \eqref{eq:LabeledVsUnlabeledInterleaving}, and thus there is a $\delta$-good map $\alpha$ that also achieves the unlabeled distance. 
\end{proof}

We conclude this section by showing that the interleaving distance is intrinsic on the space of finite (unlabeled) merge trees. 
Recall from~\cref{def:Metric_IntrinsicBottleneck}, $\hat{d}$ denotes the intrinsic metric induced by a metric $d$. 

\begin{cor}
\label{cor:intrinsicinterleaving}
For the space of finite (unlabeled) merge trees, $d_I = \hat{d}_I$.
\end{cor}

\begin{proof}
Let $T$ and $T'$ be two merge trees, and set $\delta = d_I((T,f), (T',f'))$. 
Let $\pi,\pi'$ be optimal labelings such that $d_I((T,f), (T',f')) = d_I^L((T,f,\pi), (T',f',\pi'))  = \delta$, as established by \cref{cor:OptLabel}.

Now consider the space of labeled merge trees $\LMT$. 
By \cref{cor:labelintrinsic}, there exists a geodesic $\gamma: (T,f,\pi) \rightsquigarrow (T',f',\pi')$ in $\LMT$ such that the length $L_{d_I^L}(\gamma)= \delta$. 

Note that $\gamma$ can be projected to a path $\gamma'$ from $T$ to $T'$ in the space of (unlabeled) merge trees $\MT$ by simply ignoring the labeling. 
As $d_I((T,f), (T',f')) \le d_I^L((T,f,\pi_1), (T',f',\pi_2)) $ for any labelings $\pi_1$, $\pi_2$ between any two trees $T$ and $T'$, we have
\begin{align}
\label{eqn:lengthlabeled}
\hat{d}_I(T, T') \le L_{d_I}(\gamma') \le L_{d_I^L}(\gamma) = \delta. 
\end{align}
On the other hand, by definition of the intrinsic metric $\hat{d}_I$ induced by $d_I$, 
\begin{align}\label{eqn:lengthunlabeled}
\hat{d}_I(T, T') \ge d_I(T, T') = \delta.
\end{align}
Combining equations \eqref{eqn:lengthlabeled} and \eqref{eqn:lengthunlabeled}, we conclude that $\hat{d}_I(T, T') = d_I(T,T')$ for any two merge trees $T$ and $T'$.
\end{proof}

%-------------------------
% Discussion

\section{Concluding Remarks and Discussion}
\label{sec:discussion}

In this paper, we investigated whether interleaving-type distances for (finite) labeled or unlabeled merge trees are intrinsic or not, and presented positive answers in both cases. In the case of labeled trees, the geodesic between two labeled merge trees can be characterized and computed easily, and we also showed how to compute the 1-center of a set of labeled merge trees. For unlabeled merge trees, however, computing the geodesic (even if just numerically estimating it) between two merge trees appears to be significantly harder, part of the reason being that it is NP-hard to approximate the interleaving distance between two merge trees, as pointed out in \cite{Touli2018}. 

On the other hand, a simpler and easier to compute object is the bottleneck distance $d_B(T_1, T_2)$ between two (unlabeled) merge trees. We conjecture that the intrinsic distance $\hat{d}_B$ induced by $d_B$ is in fact equivalent to $\hat{d}_I (= d_I)$. 

Another natural question is whether (some of the) results for merge trees in this paper can be extended to contour trees. 
As a first question, can we characterize and compute the midpoint (i.e., the contour tree representing the 1-center) for two labeled contour trees under either $\hat{d}_I$, $\hat{d}_B$, or $\hat{d}_{FD}$ (where we remind the reader that $d_{FD}$ denotes the functional distortion distance)? 
One idea is to compute the join and split trees of input contour trees, and compute the midpoint of the pair of join trees (resp., the pair of split trees). Note that each join or split tree can be viewed as a merge tree.
Next we need to use the common ancestor information in both trees to construct a midpoint for the two contour trees. This step could be subtle: in particular, it is known \cite{WWW14} that in general, given a descending (join) tree $T_J$ and an ascending (split) tree $T_S$ with consistent functions associated to them, there may not exist a contour tree (or even a graph) whose join and split trees are equal to $T_J$ and $T_S$, respectively. 
If such a contour tree exists, then it is unique, and the algorithm by Carr et al.~\cite{CarrSnoeyinkAxen2003} will compute this tree in near linear time. 

\update{
Finally, understanding theoretical properties of distances between merge trees has many practical implications. 
For instance, in scientific visualization, such distances may be employed to study ensemble data sets that arise from scientific simulations (e.g.,~\cite{YanWangMunch2020, PontVidalDelon2021}). 
Theorem~\ref{thm:LabeledVsUnlabeledInterleaving} suggests the potential development of computing interleaving distances between unlabeled merge trees.
Building on the work presented in this paper, Yan et al.~\cite{YanWangMunch2020} computed the structural average and geodesics of merge trees for uncertainty visualization. 
They explored various labeling strategies for computing interleaving distances between merge trees. 
Furthermore, Curry et al.~\cite{CurryHangMio2021} estimated the interleaving distance between unlabeled merge trees by searching for an optimal alignment between nodes in the trees with respect to a certain cost function; such estimation was used for classification and comparison of point cloud data.  
Moving beyond this paper, we envision a number of future applications in topological data analysis and visualization. 
}

%-------------------------

\paragraph{Acknowledgments}
Our initial research collaboration began during the Dagstuhl Seminar 17292: Topology, Computation and Data Analysis in July 2017 (organized by Bei Wang, Hamish Carr, and Michael Kerber).
We thank all members of the breakout session on Reeb graphs for stimulating discussions.
We are grateful to the Institute for Computational and Experimental Research in Mathematics (ICERM) for supporting us through the Collaborate@ICERM program in August 2018.
EM was partially supported by National Science Foundation (NSF) through grants CMMI-1800466, DMS-1800446, and CCF-1907591. 
KT was supported by an ARC Discovery Early Career fellowship.
BW was partially supported by Department of Energy (DOE) DE-SC0021015,  NSF IIS-1513616 and DBI-1661375, as well as National Institutes of Health (NIH) R01EB022876.
YW was partially supported by NSF CCF-1740761 and DMS-1547357, as well as NIH  R01EB022899. 

%-------------------------
%\update{\section*{References}}
%\bibliographystyle{abbrv}
%\bibliography{mergetrees}

%-------------------------

\end{document}